\documentclass[11pt]{article}
\usepackage[margin=1in]{geometry}

\usepackage{amsmath, amssymb, amsthm}
\usepackage[mathic]{mathtools}
\usepackage{algorithmic}
\usepackage{algorithm}
\usepackage[T1]{fontenc}
\usepackage{bbm}
\usepackage{float}
\usepackage{subcaption}
\usepackage{booktabs}
\usepackage{libertine}
\usepackage{authblk}

\usepackage{thm-restate}
\usepackage[libertine]{newtxmath} % must load after ams packages
\usepackage[scaled=0.96]{zi4} % use Inconsolata as the typewriter font

\usepackage{a4wide, dsfont, microtype, xcolor, paralist}

\usepackage[ocgcolorlinks]{hyperref} % Option ocgcolorlinks makes links only colored on screen and not on printouts
\colorlet{DarkRed}{red!50!black}
\colorlet{DarkGreen}{green!50!black}
\colorlet{DarkBlue}{blue!50!black}

\let\epsilon\varepsilon

\usepackage{xspace}

\usepackage{tikz}
\usetikzlibrary{decorations.pathreplacing}
\usetikzlibrary{plotmarks}
\usetikzlibrary{positioning,automata,arrows}
\usetikzlibrary{shapes.geometric}
\usetikzlibrary{decorations.markings}
\usetikzlibrary{positioning, shapes, arrows}

\PassOptionsToPackage{usenames,dvipsnames,svgnames}{xcolor}
\definecolor{orange}{RGB}{235,90,0}
\definecolor{darkorange}{RGB}{175,30,0}
\definecolor{turkis}{RGB}{131,182,182}
\definecolor{darkturkis}{RGB}{31,82,82}
\definecolor{green}{RGB}{102,180,0}
\definecolor{darkgreen}{RGB}{51,90,0}
\definecolor{myblue}{RGB}{0,0,213}
\definecolor{mydarkblue}{RGB}{0,0,100}
\definecolor{mybrightblue}{HTML}{74B0E4}
\definecolor{mybrighterblue}{HTML}{B3EAFA}
\definecolor{lila}{RGB}{102,0,102}
\definecolor{darkred}{RGB}{139,0,0}
\definecolor{darkyellow}{RGB}{188,135,2}
\definecolor{brightgray}{RGB}{200,200,200}
\definecolor{darkgray}{RGB}{50,50,50}
\definecolor{amaranth}{rgb}{0.9, 0.17, 0.31}
\definecolor{alizarin}{rgb}{0.82, 0.1, 0.26}
\definecolor{amber}{rgb}{1.0, 0.75, 0.0}
\definecolor{green(ryb)}{rgb}{0.4, 0.69, 0.2}
\definecolor{hanblue}{rgb}{0.27, 0.42, 0.81}
\definecolor{grannysmithapple}{rgb}{0.66, 0.89, 0.63}

\newtheorem{theorem}{Theorem}[section]

\newtheorem{observation}[theorem]{Observation}

\newtheorem{lemma-rstbl}[theorem]{Lemma}
\newtheorem{obs-rstbl}[theorem]{Observation}
\newtheorem{theorem-rstbl}[theorem]{Theorem}

\usepackage{environ}
\NewEnviron{cproblem}[1]{%
\begin{center}\fbox{\parbox{5.5in}{%
    {\centering\scshape #1\par}%
    \parskip=1ex
    \everypar{\hangindent=1em}%
    \BODY
}}\end{center}}

\title{Hole Detection and Healing in Hybrid Sensor Networks}

\author{Mansoor Davoodi}
\author{Esmaeil Delfaraz}
\author{Sajjad Ghobadi}
\author{Mahtab Masoori}
\affil{Department of Computer Science and Information Technology, Institute for Advanced Studies in Basic Sciences, Zanjan, Iran

mdmonfared@iasbs.ac.ir, \{esmaieldelfaraz, s.ghobadi1370, mahtab.masoori\}@gmail.com
}

\date{}
\begin{document}
\maketitle
\begin{abstract}
	Although monitoring and covering are fundamental goals of a wireless sensor network (WSN), the accidental death of sensors or the running out of their energy would result in \textit{holes} in the WSN. Such holes have the potential to disrupt the primary functions of WSNs. This paper investigates the hole detection and healing problems in hybrid WSNs with non-identical sensor sensing ranges. In particular, we aim to propose centralized algorithms for detecting holes in a given region and maximizing the area covered by a WSN in the presence of environmental obstacles. To precisely identify the boundary of the holes, we use an additively weighted Voronoi diagram and a polynomial-time algorithm.Furthermore, since this problem is known to be computationally difficult, we propose a centralized greedy $\frac{1}{2}$-approximation algorithm to maximize the area covered by sensors. %~\cite{wang2003portal}. %for healing problem by moving mobile sensors towards the holes and improving the coverage.
	Finally, we implement the algorithms and run simulations to show that our approximation algorithm efficiently covers the holes by moving the mobile sensors. %The simulation results confirm the efficiency of the proposed algorithms in terms of detecting and healing the holes.
\end{abstract}
% \thispagestyle{empty}
% \pagebreak
\section{Introduction}
A Wireless Sensor Network (WSN) comprises a set of sensors that are deployed over a region of interest (RoI) with the capabilities of sensing and communication with each other~\cite{li2015coverage}. Each sensor has a limited resource with low power signal processing, computing, and connectivity abilities that can sense a restricted area and performs specific tasks such as target detection, target tracking, area monitoring, etc.~\cite{davoodi2019algorithms,ghosh2004estimating}. WSNs have been used in many military and civil applications, such as battlefield surveillance, intrusion detection, environmental protection, disaster recovery, traffic control, and air pollution monitoring~\cite{davoodi2019algorithms, kalwaghe2014literature, kang2013detection, li2015coverage}. One of the main challenges in WSNs is managing energy by optimizing sensors' deployments.  

The sensors’ energy is powered by a battery installed on them, and the presence of holes in a region may mean that %It is possible some of these
some sensors run out of energy
%, which lead to hole formation 
in the network.
In several applications of WSNs, the location of sensors is determined beforehand, and sensor deployment can be performed manually. In contrast, in remote and hostile environments, a random deployment strategy is one possible solution. Due to random deployment, the distribution of the sensors is non-uniform. Consequently, the surveillance of the environment cannot be completely fulfilled~\cite{ghosh2004estimating}. It is clear that at least one alive sensor is required to cover each point within the environment in WSNs. Therefore, due to the accidental death of the sensors or running out of their energy, full coverage of an environment is difficult~\cite{li2015coverage}. Consequently, occurrence of holes within WNNs is inevitable~\cite{sahoo2014hora}. 

%The presence of holes in a region degrades the performance of WSNs, e.g., changing the network topology and decreasing data reliability and communication among the sensors~\cite{corke2007finding}, increasing data transmission load by the boundary sensors of the holes and running out of their energy (the \textit{hole diffusion} problem)~\cite{li2011geographic}.

The presence of holes in a region degrades the performance of WSNs. It changes the network topology and decreases data reliability and communication among the sensors~\cite{corke2007finding}.
It is also more likely that the sensors that are located on the boundary of the holes will transfer more information to bypass the holes. As a result, the energy of the sensors will be depleted and leads in enlarging the hole areas (\textit{hole diffusion} problem)~\cite{li2011geographic}.
%Existence of holes may indicate the features of the environment, such as fires and floods that may prevent the deployment of new sensors~\cite{kanno2009detecting}.
However, the detection of the boundary sensors of the holes helps in improving the performance of WSNs because %\textcolor{blue}{(\textit{i}) 
%it results in calculating the amount of sensing data inside the holes}~\cite{tong2006discovering}; 
(\textit{i}) it provides a healing approach to maximize coverage; and (\textit{ii}) it addresses to hole diffusion problem and obtains the boundary sensors of the holes. %Motivating by these reasons, in this study

\paragraph{Our contribution.} %In this study, we focus on finding and covering the holes in a region. In the first step, we design an algorithm to find the holes precisely in the presence of obstacles in the environment, using the additively weighted Voronoi diagram. Then, we present a simple greedy algorithm with a $1/2$-approximation factor to maximize the coverage area. Finally, we run simulations and observe that our simple approximation algorithm covers the environment efficiently in practice.
In this study, we focus on computing and covering the holes in a WSN. In particular, we investigate the problems of hole detection and healing in a network of hybrid sensors with non-identical sensing ranges. First, we use the \textit{Additively Weighted Voronoi Diagram} (AWVD) to identify the holes in the absence of obstacles. We propose a centralized algorithm to report the boundary edges of the holes. The time complexity of the algorithm  is $O(n \cdot \log^2 n)$, where $n$ is the total number of sensors. We extend this strategy to the case where the environment consists of a set of obstacles, and propose an algorithm to detect the holes in $O(n (\log^2 n + l \cdot z))$ time, where $z$ and $l$ are the total numbers of vertices and the number of obstacles, respectively.
Furthermore, we investigate the problem of healing the discovered holes and maximizing the coverage area by moving the available sensors. As this problem is known to be NP-hard~\cite{wang2003portal}, we present a centralized greedy $\frac{1}{2}$-approximation algorithm to maximize the coverage area. Finally, we run simulate the algorithms and show that our hole detection algorithms identify the boundary points of the sensors exactly. Also, the proposed simple approximation algorithm for the hole healing problem covers the environment efficiently in practice.
 
We present the related work in Section~\ref{related work} and introduce the preliminaries and formal definition of the problem in Section~\ref{Problem Statement}. In Section~\ref{Hole Detection Algorithms}, we propose a centralized algorithm to detect the holes in the absence and presence of holes. %afor hole detection in a case that there is no obstacle in the environment, and then extend the algorithm for the case that the environment consists of a set of obstacles. 
In Section~\ref{Hole Healing}, we propose an approximation algorithm to cover the holes. In Section~\ref{sec : numerical tests}, we evaluate the performance of the presented algorithms by simulating them, and finally, we conclude in the last section.

\section{Related Work}\label{related work}
%We review the works that have considered hole detection and healing problems in WSNs in more detail.
We review the literature related to hole detection and healing problems.

The first line of papers used the Voronoi diagram and triangulation that are two critical structures in computational geometry~\cite{aurenhammer1991voronoi, fortune1995voronoi}.
%that is widely used for hole detection and healing problems in WSNs~\cite{fortune1995voronoi,aurenhammer1991voronoi}.
Wang et al.~\cite{wang2003portal} proposed a distributed bidding protocol in hybrid sensor networks. In this protocol, static sensors detect coverage holes and estimate their size using the Voronoi diagram and bid for mobile sensors to move towards the farthest Voronoi vertex. Also, they proved NP-hardness of coverage problem. Ghosh~\cite{ghosh2004estimating} and Wang et al.~\cite{wang2006movement} proposed Voronoi-based algorithms to calculate the exact amount of coverage holes under random deployment. Ghosh~\cite{ghosh2004estimating} studied hole detection and healing problems in a hybrid network in which each mobile sensor has the same sensing range. In this study, after estimating the number of coverage holes, the aim is to calculate the number of additional mobile sensors and their locations to maximize the environment's coverage. Wang et al.~\cite{wang2006movement} presented three-hole healing distributed algorithms VEC, VOR, and Minimax. VEC pushes sensors away from each other to create maximum coverage in each Voronoi cell. VOR, which is a greedy algorithm, moves a sensor towards its farthest Voronoi vertex. Minimax selects a point as a Minimax point, i.e., the point whose distance from the farthest Voronoi vertex is minimal and moves the sensors. Babaie and Pirahesh~\cite{babaie2012hole} used triangulation to identify coverage holes and estimate their size. They also determined the required number of mobile sensors and their target location heal the holes. Li and Zhang~\cite{li2015coverage} proposed a distributed algorithm to evaluate whether coverage holes exist in WSNs. They used Delaunay triangulation to indicate the topology of a network and identify the hole boundary sensors. Zhang et al.~\cite{zhang2013energy} proposed a centralized algorithm for hole detection and healing problems in a hybrid sensor network. They used Delaunay triangulation of the static sensors to determine holes among each triangle. In the hole healing algorithm, they identified the place where the mobile sensors should cover first. After moving one mobile sensor, the algorithm updates the triangulation and iterates the process until the coverage percentage reaches a certain threshold. Wu et al.~\cite{wu2007delaunay} presented a centralized method to maximize coverage in a given environment with obstacles. They assumed a predefined number of sensors with a probabilistic sensor detection model. First, they deployed new sensors to eliminate coverage holes around the boundary of the environment and obstacles. Then, they applied a deterministic sensor deployment method based on Delaunay triangulation and improved the coverage of the environment. The time complexity of their algorithm is $O(n^2 \cdot \log n)$, where $n$ is the number of sensors. Alablani and Alenazi~\cite{alablani2020edtd} proposed a sensor deployment method, called the Evaluated Delaunay Triangulation-based Deployment for Smart Cities, in the presence of obstacles. The proposed method was designed for IoT smart cities that use Delaunay triangulation and k-means clustering to determine the sensor locations to improve the coverage area. Ghahroudi et al.~\cite{ghahroudi2020voronoi} designed a distributed Voronoi-based Cooperative Node Deployment algorithm. In this algorithm, first, each node constructs its Voronoi cell. Subsequently, the covered and hole areas are identified in each Voronoi cell. Finally, mobile sensors move to a proper location within their Voronoi cells to increase the coverage area. Koriem and Bayoumi~\cite{koriem2018detecting} proposed a hole detection algorithm to calculate the hole area and minimize the energy consumption. First, the algorithm partitions the area into a set of grid cells. Then, it identifies the nearest nodes to each cell by comparing the cell's position with each node laying in that cell. Finally, it computes the hole area by using the triangulation method.

We proceed by reviewing more works. Fekete et al.~\cite{fekete2004neighborhood} proposed a boundary detection algorithm based on a statistical approach. Their idea is that the sensors on the holes' boundary have fewer average degrees than the sensors are within the network. Their algorithm has complex computations and needs high sensor density. Ghrist and Muhammad~\cite{ghrist2005coverage} and Silva et al.~\cite{de2005blind} studied the hole detection problem using homology theory and algebraic method. They proposed centralized coordinate-free algorithms based on the sensors' connectivity information to detect the boundary of the holes. The algorithm's time complexity is $O(n^5)$~\cite{ghrist2005coverage}. Unfortunately, it does not provide any guarantee to detect the hole boundary accurately. The time complexity of the proposed algorithm in~\cite{de2005blind} is $O(q \cdot v^2)$, where $v$ is the maximum number of active sensors that overlap a sensor's sensing area and $q$ is the maximum number of redundant sensors in a hole. Their method does not perform well in dense networks. 

Kang et al.~\cite{kang2013detection} proposed a decentralized sensor-based algorithm that obtains the holes' boundary points. The hole healing algorithm is based on the perpendicular bisector of the boundary points. The healing sensors are deployed on hole boundary bisectors to create full coverage in the environment. The time complexity of the algorithm is $O(a \cdot f \cdot n)$, where $f$ is the maximum number of sensors that are neighbouring to a sensor ($f < n$), and $a$ is the maximum number of boundary points in the network. Sahoo et al.~\cite{sahoo2010vector} proposed a distributed hole detection algorithm that finds the boundary of uncovered regions. They designed a hole recovery $O(n^3)$ time algorithm to identify the mobile sensors' magnitude and direction using a vector method.

 Ma et al.~\cite{ma2011computational} proposed a distributed $O(f \cdot n)$ time algorithm such that each sensor detects the holes based on 1-hop and 2-hop neighbours. Kun et al.~\cite{bi2006topological} proposed a distributed hole detection algorithm using the communication graph of the sensors. Each sensor is identified as a hole boundary by comparing its degree and the average degree of its 2-hop neighbours. The algorithm cannot precisely detect all boundary sensors. Funke~\cite{funke2005topological} presented a distributed coordinate-free algorithm to identify holes based on the communication topology of the sensors. The time complexity of the algorithm is $O(n^5)$. Funke and Klein~\cite{funke2006hole} described a graph-based algorithm to identify the boundary sensors using the unit disk graph model. The algorithm requires a high sensor density.

Kalwagh and Dusane~\cite{kalwaghe2016hole} and Senouci et al.~\cite{mellouk2013localized} proposed distributed algorithms to identify and heal the holes in hybrid and mobile sensor networks, respectively. In these studies, the sensing range of the sensors is identical, and there is no obstacle in the environment. They detected the boundary sensors of the holes by using the TENT rule. The hole healing algorithm is based on virtual forces that move the sensors to locate at a suitable distance from the hole. The time complexity of these algorithms is $O(n^2)$. Also, the proposed algorithm in~\cite{mellouk2013localized} cannot detect the holes on the boundary of the network.

Zou and Chakrabarty~\cite{zou2003sensor} proposed a virtual force algorithm to enhance the coverage of an environment in the presence of obstacles.
The algorithm uses a combination of attractive and repulsive forces between the sensors and obstacles to identify the sensors' final position that guarantees maximum coverage. Chang and Wang ~\cite{4357431} proposed a method to rearrange sensor nodes to maximize the area covered by the sensors in an environment containing a set of obstacles. In this method, the sensors' and obstacles' density is calculated, and a sensor moves to a position with low density.

Priyadarshi and Gupta1~\cite{priyadarshi2020coverage} addressed the hole detection problem by identifying the coverage hole in the monitoring region and moving sensors toward the holes to increase the coverage. Their algorithm is a heuristic method without any guarantee on the coverage area obtained after moving the sensors. El Khamlichi et al.~\cite{el2018recovery} proposed a coverage hole detection and recovery algorithm that detects the coverage holes and redundant nodes based on the distance between the nodes. Then heals the coverage holes using a gradient algorithm by moving the redundant nodes. Khalifa et al. ~\cite{khalifa2020distributed} presented a distributed hole detection and repair algorithm using mobile sensors with two phases. First, they detect the hole by estimating the size and
position of the existing coverage holes. Second, the hole repair is applied to restore coverage holes using the distance and movement direction of nodes selected for the hole repair phase. Wu~\cite{wu2020efficient} proposed a hole recovery method, called Neighbor node Location-based Coverage Hole Recovery, that adds mobile sensor nodes at an optimal location using the connectivity of 1-hop neighbour nodes to recover coverage holes.

Khedr et al.~\cite{khedr2020coverage} designed a coverage-aware and planar face topology structure (CAFT). 
%CAFT consists of two phases. 
CAFT partitions the nodes into groups, called faces, using the minimum number of nodes to achieve maximum coverage and connectivity while putting the other nodes in sleep mode. They also proposed a target tracking algorithm that runs on the CAFT and evaluates its performance. Hu et al.~\cite{hu2020coverage} studied hole detection and healing problems in the presence of mixed and mobile sensors. They calculated the location of nodes and obtained the coverage holes using the continuous
maximum flow method. Further, they proposed a healing method that moves mobile sensors to new locations to improve the coverage using the harmony search algorithm.

Singh and Chen~\cite{singh2020sensing} designed a chord-based coverage hole identification method to identify both open and closed holes by determining the hole boundary nodes. The complexity of the method is $O(n^3 f^3 + \rho \theta^2)$, where $\rho$ is the number of identified cycles from all the hole chords, and $\theta$ being the
maximum number of hole vertices of each cycle. They also presented a covering method that completely covers the network using the minimum number of nodes. The covering method's complexity is $O(1.44^{\lambda} + \lambda^2 d^2)$, where $\lambda$ is the number of hole boundary nodes and $d$ is the maximum number of nodes that are
neighbours of the hole boundary nodes. 

As far as the authors know, the problem of hole detection and healing for a set of sensors with different sensing ranges in an environment with obstacles has not been studied yet. On the other hand, as aforementioned, most of the hole detection and healing algorithms in the literature have considered a network of identical sensors. Also, they have considered an environment without obstacles. The drawback of most of these algorithms is that they cannot report the boundary of holes accurately. Further, most of the algorithms have high complexity time and are not efficient in a sensor network with different sensing ranges.

\section{Problem Statement} \label{Problem Statement}
We consider a wireless sensor network including static and mobile sensors with different sensing ranges deploying randomly over an environment with obstacles.
% Without loss of generality, assume that another sensor completely covers no sensor. Our goal is to present a centralized algorithm for detecting and calculating the area of the holes. In the next step, we aim to maximize coverage in the environment using the mobile sensors' local mobility.
Following Tan et al.~\cite{tan2013approximate}, we suppose that the environment contains transparent obstacles, where the sensors do not lie on them, but their sensing signals can penetrate through them. The environment and obstacles are modelled as polygons, and the boundary information of the obstacles and environment, coordinates of the sensors and their sensing ranges are given in advance. Formally, let $O = \{o_1, \ldots, o_l\}$ be a set of obstacles, $SS =\{ss_1, \ldots, ss_t\}$ be a set of static sensors (i.e., we do not or cannot move them from their initial positions) and $MS = \{ms_1, \ldots, ms_m\}$ be a set of mobile sensors (i.e., we can move them from one position to another one, without colliding with the environment) in an environment. The goal is to detect the set of holes $H = \{H_1, \ldots, H_{|H|}\}$, where $|H|$ is the number of holes and $H_i \in H$, for $i=1,2,\dots,|H|$, represents an area (polygon) where no sensor covers it. We call this problem \textit{the hole detection problem}. Also, for the sake of well-definition, we assume $H \cap O= \emptyset$. 

To represent the holes, we utilize a set of boundary points as $B = \{h_1^1, \ldots, h_1^{|h_1|}, \ldots,h_{|H|}^1, \ldots, h_{|H|}^{|h_{|H|}|} \}$, where $|h_i|$ is the number of intersection (boundary) points that makes the hole $H_i$. The intersection points of sensor $s_i$ with (\textit{i}) its neighbor sensors, (\textit{ii}) obstacles, (\textit{iii}) and the boundary of the environment that are not covered by other sensors are called \textit{boundary points} of $s_i$. We compute the boundary edges $H_i = \{e_{i1}, \ldots, e_{i|h_i|}\}$, for $i = 1, \ldots, |H|$ in which $e_{ij}$ ($j = 1, \ldots, |h_i|$) indicates the edges of $H_i$. After detecting the holes, we solve \textit{the healing problem}, that is, maximizing the covered area by moving the mobile sensors. %In the upcoming section, we investigate the hole detection problem. %After finding the holes, we present a greedy hole healing algorithm by moving mobile sensors. The proposed algorithm is a centralized greedy $1/2$-approximation that moves mobile sensors through the holes and maximizes coverage in the environment.
\section{Hole Detection}\label{Hole Detection Algorithms}
    %We start with the hole detection problem where 
    In this section, the goal is to present centralized algorithms in order to identify the boundary points of the existing holes in a given environment.
	%In this section, we introduce a centralized algorithm to solve the hole detection problem. 
	First, we consider the case that the environment does not contain any obstacle.
	%, and propose the algorithm. 
	Then, we extend the results to the case that the environment includes some obstacles.

\subsection{Hole Detection in an Environment without Obstacle}
	Let $S =\{s_1, \ldots, s_n\}=SS \cup MS$ be the set of all sensors including both static and mobile sensors, where $s_i = (c_i, r_i)$ is a disk centered at $c_i$ with radius $r_i$. To find the boundary points of the holes, we use the additively weighted Voronoi diagram. Recall that, given $n$ points (here they are sensors) in the plane (or environment), AWVD partitions the environment into $n$ Voronoi cells (one for each sensor) with the property that a point $p$ lies in the Voronoi cell corresponding to $s_i$ if and only if $d(c_i, p) - r_i \le d(c_j, p) - r_j$ for any $s_j \in S$, where $d(\cdot, \cdot)$ is the distance function. The Voronoi edges of AWVD can be line segments or arcs of hyperbolas~\cite{ash1986generalized} (see Figure~\ref{awvd}). 	Accordingly, we have the following observations.
	
		\begin{figure}[H]
				\centering
				\includegraphics[width=6cm]{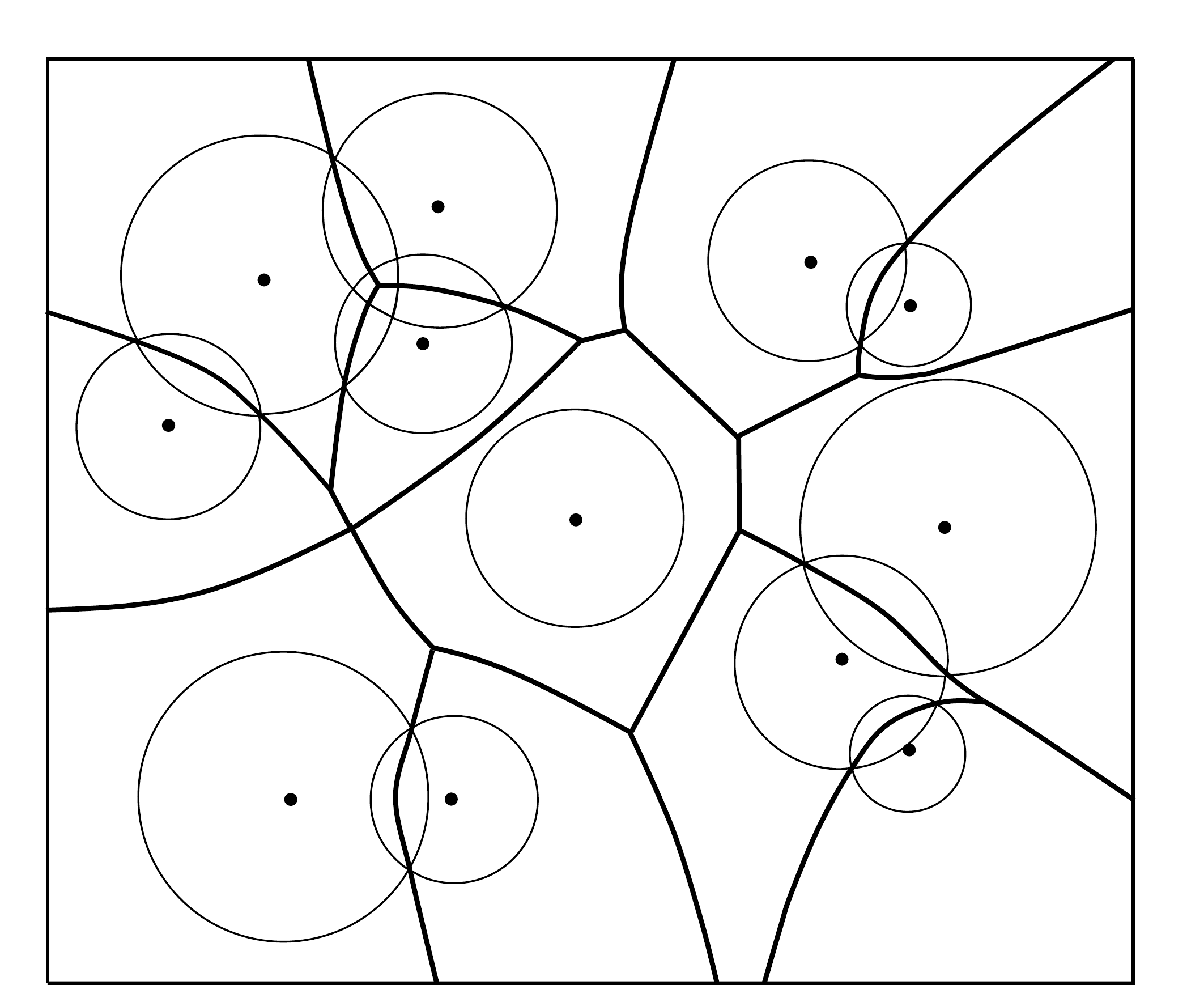}
				\caption{The additively weighted Voronoi diagram for a set of sensors.}\label{awvd}
		\end{figure}

	\begin{observation} \label{obs1}
		The intersection points of any two distinct sensors $s_i$ and $s_j$ lie on AWVD edges.
	\end{observation}
	\begin{observation}\label{obs2}
		Let $b_{ij}$ be the set of intersection points of two sensors $s_i$ and $s_j$. If some intersection point $p \in b_{ij}$ is covered by a sensor $s_k \neq s_i \neq s_j$, then $p$ does not lie on the edges of AWVD (see Figure~\ref{fig7}). If there is a point $p' \in b_{ij}$ that lies on the boundary of the other sensors, then $p'$ is the Voronoi vertex based on the AWVD's properties (see Figure~\ref{fig8}).
	\end{observation}

\begin{figure}[H]
	\centering
	\begin{subfigure}[b]{.5\linewidth}
		\centering
		\includegraphics[scale=.5]{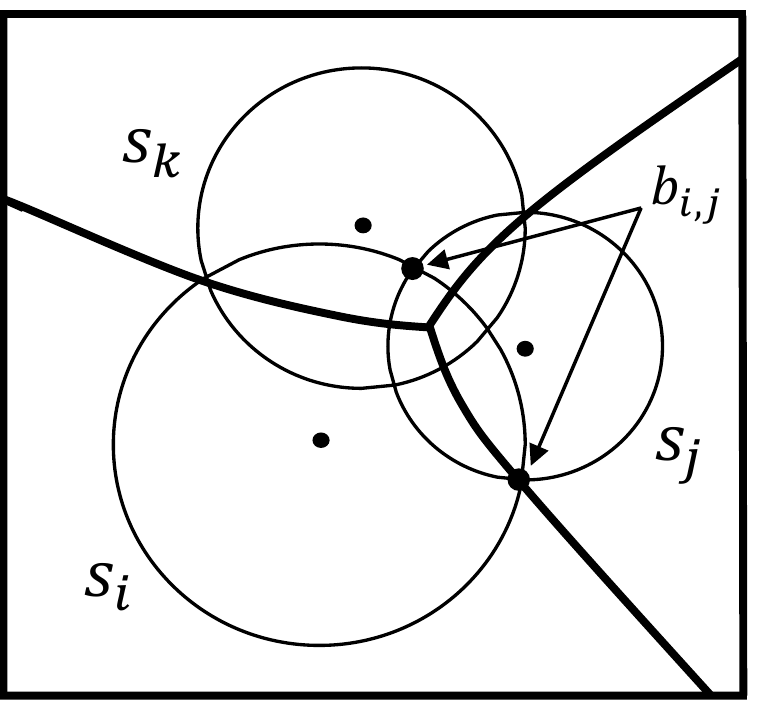}
		\caption{}
		\label{fig7}
	\end{subfigure}
	\begin{subfigure}[b]{.45\linewidth}    
		\centering
		\includegraphics[scale=.52]{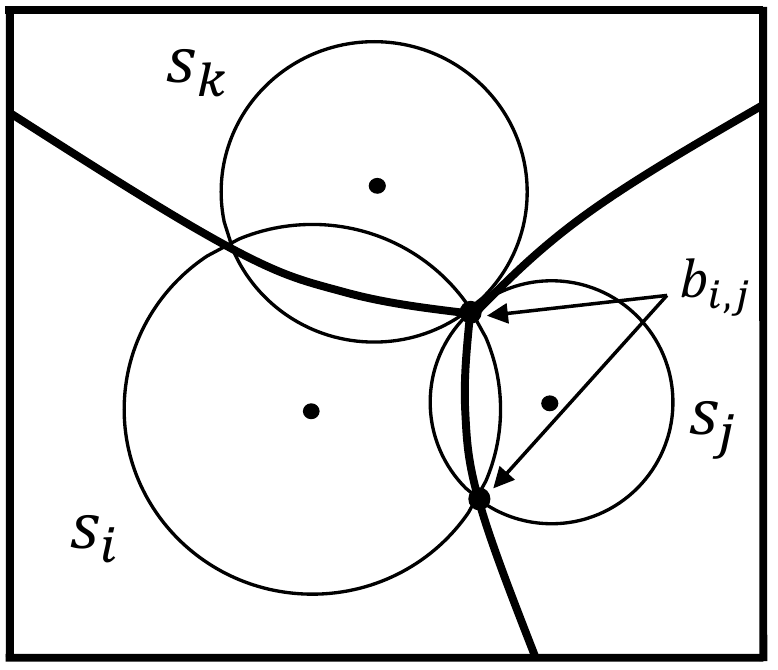}
		\caption{}
		\label{fig8}
	\end{subfigure}
	\caption{(a) Sensor $s_k$ covers at least one of the intersection points of the sensors $s_i$ and $s_j$, and (b) there is at least a point in $b_{ij}$ that lies on the boundary of the sensor $s_k$.} 
	\label{fig9} 
\end{figure}
Observations~\ref{obs1} and~\ref{obs2} imply that the intersection points of sensors that are not covered by other sensors either lie on the Voronoi edges or are a subset of Voronoi vertices. Indeed, such intersection points are known as the hole boundary points. 
	
\begin{observation}
    Let $S$ be the set of $n$ sensors in the environment. The total number of boundary points of the holes is $O(n)$.
\end{observation} \label{theorem1}
\begin{proof}
    Observe that the AWVD of $n$ sensors contains $O(n)$ vertices and edges~\cite{aurenhammer1991voronoi}. Based on Observations~\ref{obs1} and~\ref{obs2}, the number of boundary points that lie on every Voronoi edge is at most 2, then the number of boundary points of the holes is $O(n)$.
\end{proof}
	Recall that there are two types of holes in the environment: (i) \textit{close hole}: the boundaries of the holes that are formed by the sensing area of the sensors, and (ii) \textit{open hole}: the boundary of the holes that are formed by the sensing area of the sensors and boundary of the environment or obstacles.
	
	Using Observations~\ref{obs1} and~\ref{obs2}, AWVD enables us to detect all of the boundary points, represented by $B$, precisely. To do so, we develop an incremental algorithm that constructs the boundary edges of each hole using the set $B$ in two steps. We term this algorithm HDAO stands for Hole Detection in the Absence of Obstacles. 
	First, it sorts the boundary points of each sensor in clockwise order and computes the holes' boundary edges.
	Second, for each hole $H_i$, HDAO outputs a set of boundary edges constructing that hole. We first specify the details of the first step and then proceed to analyze the second step.
	
	\paragraph{First step.} Let $s_i$ be the $i$th sensor and $B_i = \{h_{i1}, \ldots, h_{iu}\}$, $u \le n$, be the set of boundary points that lie on the boundary of $s_i$ which is sorted in clockwise order. Here the objective is to start from $h_{i1}$ and march along $s_i$ until reaching $h_{i1}$ again, and reporting the boundary edges of the holes.
	More precisely, consider two consecutive points $h_{ij}$ and $h_{i(j+1)} \in B_i$ that are created by the intersection of $s_i$ with $s_k$ and $s_v$, respectively. If the line segment $h_{ij} h_{i(j+1)}$ is not contained in $s_i \cap s_k$ and $s_i \cap s_v$, the algorithm considers the line segment $h_{ij} h_{i(j+1)}$ as a boundary edge of holes. HDAO repeats this process for all the sensors to determine all the boundary edges of holes.
	
    \paragraph{Second step.} Now, let $E$ be the set of all the boundary edges of holes. First, the algorithm sorts $E$ according to the $x$-coordinate of the starting point of the edges.
    Then, HDAO starts from an arbitrary edge $e_{ij}$ and selects edge $e_{i(j+1)}$ in clockwise order such that $e_{i(j+1)}$ does not lie on the boundary of the sensor that edge $e_{ij}$ lies on it, also, $e_{ij}$ and $e_{i(j+1)}$ shares an endpoint (see Figure~\ref{hole}). The algorithm iterates this process until it reaches the starting point of the first edge and visits all the edges. Note that to calculate the holes' area, HDAO follows the method proposed by Babaie and Pirahesh~\cite{babaie2012hole}.

	\begin{figure}[H]
		\centering
		\includegraphics[width=6cm]{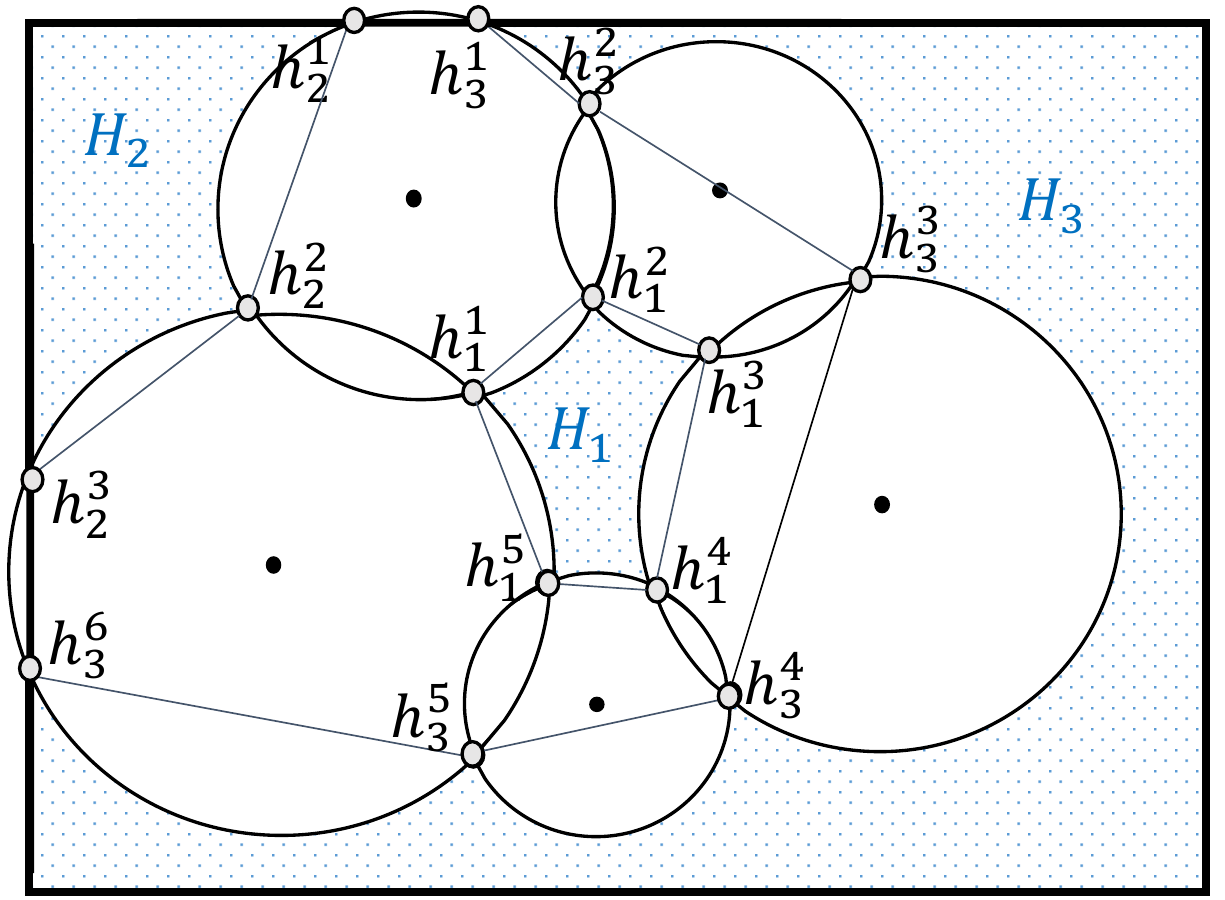}
		\caption{$H_1$ is a close hole, $H_2$ and $H_3$ are open holes, and $h_i^j$ denotes the $j$th boundary point of $i$th hole.}\label{hole}
	\end{figure}
	
	Procedure details of the algorithm are given below.
	
	\begin{algorithm}
		\caption{HDAO(S, W)}
		\label{Forward}
		\hspace*{\algorithmicindent} \textbf{Input:} An environment $W$ containing a set of sensors $S$.\\
		\hspace*{\algorithmicindent} \textbf{Output:} $H=\{H_1, \ldots, H_{|H|}\}$ the set of holes in $W$.
		\begin{algorithmic}[1]\baselineskip=14pt\relax
			\STATE Construct AWVD of $S$.
			\FOR {each $s_i \in S$}
			  \STATE Sort the boundary points of $s_i$ in clockwise order and obtain the boundary edges.
			\ENDFOR
			\STATE Sort the boundary edges based on their $x$-coordinates and obtain the boundary edges of each hole.
			\STATE Triangulate the holes based on the boundary points and examine the area of each coverage hole. 
			
		\end{algorithmic}\label{HDAO}
	\end{algorithm}

	 We move to the complexity analysis. HDAO takes $O(n \cdot \log^2 n)$ time for constructing AWVD in line 1~\cite{sharir1985intersection}. The calculations of line 2 to line 4 take $O(n \cdot \log n)$ time. Therefore, the proposed algorithm's total time complexity is $O(n \cdot \log^2 n)$ time.
	
\subsection{Hole Detection with Obstacles}
	In this subsection, we extend the presented HDAO for a given environment containing a set of obstacles $O = \{o_1, \ldots, o_l\}$. We assume that there is no intersection between the obstacles. As a warm-up, we first give an observation that provides a lower bound on the complexity of the hole detection problem in the presence of obstacles.
	\begin{observation}
		A trivial lower bound on the time complexity of hole detection problem in an environment containing a set of obstacles with $z$ vertices and $n$ sensors is $\Omega(n \cdot z)$.
	\end{observation}

\begin{proof}
    
	Consider Figure~\ref{lower_bound}, each edge of an arbitrary obstacle $o_i$ intersects all $n$ sensors in the worst-case. Based on AWVD, the number of the sensors' boundary points is $O(n)$ and the number of vertices of the obstacles is $z$, so we need to report $\Omega(n \cdot z)$ boundary points to identify the boundaries of holes. As a result, the time complexity of any algorithm for hole detection problem in the presence of obstacles is at least $\Omega(n \cdot z)$.
\end{proof}

	\begin{figure}[H]
		\centering
		\includegraphics[width=8cm]{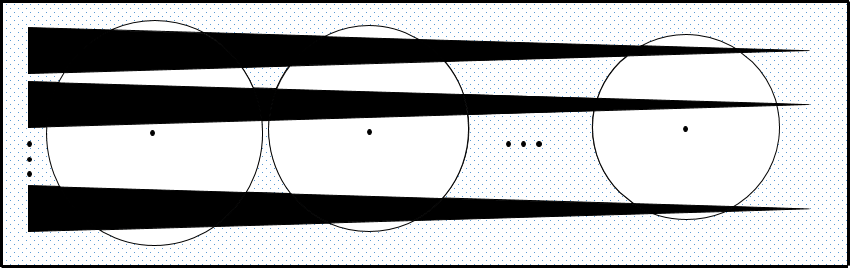}
		\caption{A lower bound for hole detection problem.}\label{lower_bound}
	\end{figure}

%\subsubsection{Computing Holes in the Presence of Obstacles}
In what follows, we propose a centralized algorithm to detect the holes when the given environment contains some obstacles. Our centralized algorithm for hole detection and boundary points recognition includes two steps: (1) ignoring the obstacles: %hole detection by using  Algorithm~\ref{HDAO}: 
in this step, we suppose that there are no obstacles in the environment and compute the boundary of each hole using Algorithm~\ref{HDAO}; (2) obstacles insertion: this step is designed to update the holes simply by inserting the obstacles one by one (see Figure~\ref{hole_obstacle}). We call this algorithm HDPO, which stands for Hole Detection in the Presence of Obstacles.
	
	\begin{figure}[H]
		\centering
		\includegraphics[width=6cm]{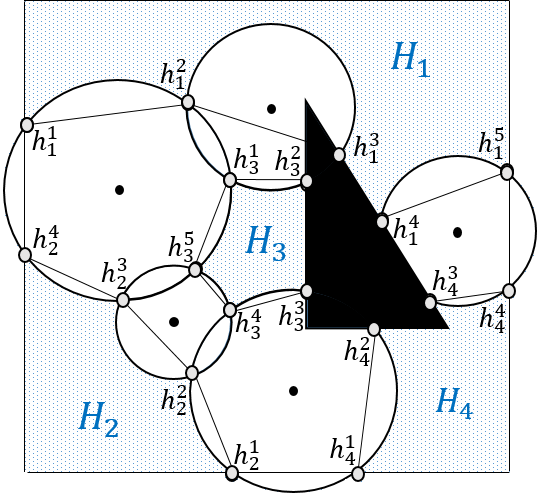}
		\caption{$H_1$, $H_2$, and $H_4$ are open holes, $H_3$ is a close hole, and $h_i^j$ is the $j$th boundary point of $i$th hole.}\label{hole_obstacle}
	\end{figure}

		\begin{algorithm}\label{HDPO}
		\caption{HDPO(S, W, O)}
		\label{Forward}
		\hspace*{\algorithmicindent} \textbf{Input:} An environment $W$ containing a set of sensors $S$ and a set of obstacles $O= \{o_i, \ldots, o_l\}$.\\
		\hspace*{\algorithmicindent} \textbf{Output:} $H=\{H_1, \ldots, H_{|H|}\}$ the set of holes in $W$.
		\begin{algorithmic}[1]\baselineskip=14pt\relax
			\STATE $H^{'}\leftarrow HDAO(S, W).$

			\FOR {each obstacle $o_i \in O$}
			\STATE Insert $o_i$ in the environment incrementally and update the set of holes $H^{'}$.
			\ENDFOR
			\STATE Report $H$.
			
		\end{algorithmic}
	\end{algorithm}
	
	\begin{theorem}[\cite{o1998computational}] \label{Theorem2}
		The intersection, union, or difference of two polygons with a total of $\nu$ vertices, whose edges intersect in $w$ points, can be constructed in $O(\nu \cdot \log \nu + w)$ time and $O(\nu)$ space.
	\end{theorem}

	Recall that some boundary edges of holes are the arcs of circles (i.e. sensors). Obviously, Theorem~\ref{Theorem2} is true for polygons with the arcs of circles.
	
	As in HDPO, the time complexity of Step 1 is $O(n \cdot \log^2 n)$. According to Theorem~\ref*{Theorem2}, the time complexity of Step 2 for hole $H_1^{'}$ is as follows:
	\begin{itemize}
		\item $O(n_1 \cdot z_1)$: updating $H_1^{'}$ after adding obstacle $o_1$ in the environment, where $z_1$ is the number of vertices of $o_1$, and $n_1$ is the number of vertices of $H_1^{'}$.
		\item $O(n_1 (z_1 + z_2))$: updating $H_1^{'}$ after adding obstacles $o_1$ and $o_2$ in the environment, where $z_2$ is the number of vertices of $o_2$.
		\item $O(n_1 (z_1 + \ldots + z_l))$: updating $H_1^{'}$ after adding obstacles $o_1, \ldots, o_l$  in the environment, where $z_l$ is the number of vertices of $o_l$.
	\end{itemize}
	So, we have:
	\begin{align*}
		O(n_1 \cdot z_1) + O(n_1 (z_1+ z_2)) + \ldots + O(n_1 (z_1+ z_2+ \ldots + z_l)) = O(n_1 \cdot l \cdot z).
	\end{align*}
	Therefore, the updating of $H_1^{'}$ takes $O(n_1 \cdot l \cdot z)$ time, where $l$ is the number of obstacles. So, the time complexity of this step for $H^{'} = \{H_1^{'} + \ldots, H_{|H^{'}|}^{'}  \}$ is:
	\begin{align*}
		O(n_1 \cdot l \cdot z) + \ldots + O(n_{|H^{'}|} \cdot l \cdot z) = O(n \cdot l \cdot z).
	\end{align*}
	Thus, the complexity of HDPO is bounded by $O(n (\log^2 n + l \cdot z))$.

\section{Hole Healing}\label{Hole Healing}
	In this section, we focus on the problem of hole healing, in which the goal is to maximize the coverage of a given environment by moving mobile sensors through the holes.
%\subsection{A Centralized Greedy Approximation Algorithm for Hole Healing}
	More formally, in a given environment, there exist a set of mobile sensors $MS = \{ms_1, \ldots, ms_m\}$ with different sensing ranges and a set of holes $H =\{H_1, \ldots, H_{|H|}\}$. We aim to maximize the coverage space by moving the mobile sensors and recovering the holes. As this problem is NP-hard~\cite{wang2003portal}, in what follows, we propose a simple greedy approximation algorithm to cover the environment. We term this algorithm Greedy-HHP. In each step, Greedy-HHP moves one mobile sensor through the hole(s), which maximizes the coverage in that step. After moving each sensor, it updates the area of holes and iterates this procedure for the rest of the sensors.
	
	\begin{algorithm}\label{HHP}
		\caption{Greedy-HHP}
		\label{Forward}
		\hspace*{\algorithmicindent} \textbf{Input:} An environment $W$ containing a set of sensors $S$ and a set of holes $H =\{H_1, \ldots, H_{|H|}\}$.\\
		\hspace*{\algorithmicindent} \textbf{Output:} New location of the mobile sensors.
		\begin{algorithmic}[1]\baselineskip=14pt\relax
			\FOR {$i=1$ to $m$}
			\STATE Move the mobile sensor $ms_i$ through the holes that cause maximum coverage in the environment.
			\STATE Update the area of the holes.
			\ENDFOR
			\STATE Report the new location of the mobile sensors.
			
		\end{algorithmic}
	\end{algorithm}
	
	\begin{theorem}
		The approximation ratio of Greedy-HHP is $1/2$.
	\end{theorem}
	\begin{proof}
		According to Greedy-HHP, in $i$th step, sensor $ms_i$ covers the hole(s), which causes the maximum coverage in that step. So, the area is covered by $ms_i$ is at least as much as the area covered by moving some sensors $ms_1^{'}, \ldots ms_k^{'}$ through the same hole by an optimal algorithm for some $k$. Precisely, let $A$ be the amount of area covered by $ms_i$ in the Greedy-HHP algorithm and $B$ be the amount of area covered by $ms_i$ in the optimal algorithm. Also, let $C$ be the total covered area by $ms_1^{'}, \ldots, ms_k^{'}$ in Greedy-HHP. Clearly $B \le A + C$. Therefore, the ratio of Greedy-HHP to the optimal algorithm is $\alpha= \frac{A+C}{B+A}$ for the mentioned sensors, and considering the greedy manner of Greedy-HHP $ \frac{A+C}{(A+C)+A}> \frac{1}{2}$. Since this inequality holds for any set of sensors, the proof is complete.
	\end{proof}
\section{Experiments} \label{sec : numerical tests}
    In this section, we report on an experimental study and evaluate the performance of our proposed algorithms. We carry out the two-stage simulation. In the first stage, we compare HDAO with that of Ma et al.~\cite{ma2011computational}, in terms of boundary holes detection. In the second stage, we point out the time complexity, accuracy and completeness of HDAO, HDPO, and the Greedy-HHP in terms of hole detection and healing.
%Several directions of future work are conceivable. First, designing algorithms that would provide a better approximation factor would be a worthwhile contribution. Second, it would be interesting to design heuristics algorithms to cover the holes efficiently.
	%\subsection{Simulation Setup}
	We simulated the algorithms using MATLAB 2015 on a PC with Intel Core i7-3770 CPU at 3.40 GHz, 8 GB RAM, Windows 8 64-bit OS for a set of mobile and static sensors with different sensing ranges. The environment consists of a set of obstacles, and sensors are deployed randomly over an RoI. The sensing range of the sensors varies from $5m$ to $20m$.

\subsection{Comparison with HDAO's Results}
	Most of the previous algorithms have been proposed for the case that the sensing range of the sensors is identical and there is no obstacle in the environment. So, in this section, to be able to compare and report the performance of our proposed algorithm, we simulate the HDAO algorithm in the absence of the obstacles where the sensing range of the sensors is $5m$. Sensors are deployed randomly over an RoI of $200m \times 200m$. We vary the number of sensors from 100 to 300 and compare HDAO with the proposed algorithm by Ma et al.~\cite{ma2011computational} in terms of running time and accuracy in detecting the boundary sensors of the holes. The results are reported in Table~\ref{table_ma}. We can see that HDAO reports the boundary sensors of the hole exactly, while the proposed algorithm by Ma et al.~\cite{ma2011computational} has some errors as mentioned in their study. Also, by increasing the number of sensors, the running time of algorithms increases. Non-surprisingly, HDAO is faster than that of Ma et al.~\cite{ma2011computational} and increasing the number of sensors does not affect the efficiency of HDAO.
	%According to Figure~\ref{comparision_ma}, HDAO reports the boundary sensors of the hole exactly, while the proposed algorithm by Ma et al. has some errors as mentioned in their study.	
\begin{table}[H]
 \centering
 \begin{tabular}{ | c | c| c | c | } 
    \hline
    Algorithm & Number of sensors & Number of boundary sensors of the holes &  Running time(s)\\ 
    \hline
     & 100 & 93 & 174\\
     HDAO & 200 & 187 & 705\\ 
     & 300 & 279 & 1364\\
    \hline
     & 100 & 79 & 300 \\ 
    Ma et al. & 200 & 160 & 992 \\
     & 300 & 242 & 2322 \\ 
    \hline
\end{tabular}
\caption{Boundary sensors detection by HDAO and the algorithm in Ma et al.~\cite{ma2011computational}.}
\label{table_ma}
\end{table}

%	\begin{figure}[H]
%		\centering
%		\begin{subfigure}[b]{.5\linewidth}
%			\centering
%			\includegraphics[scale=.23]{figure/666.png}
			%\caption{}
			%\label{fig7}
%		\end{subfigure}
%		\begin{subfigure}[b]{.45\linewidth}               
%			\centering
%			\includegraphics[scale=.16]{figure/77.png}
			%\caption{}
			%\label{fig8}
%		\end{subfigure}
%		\caption{Boundary sensors detection by HDAO and the algorithm in Ma et al.~\cite{ma2011computational}.} 
%		\label{comparision_ma} 
%	\end{figure}
	\subsection{Results on HDAO, HDPO and Greedy-HHP}
	This stage contains two parts. In the first part, we vary the sensing range of the sensors and obtain the hole boundary edges within the RoI in the presence and absence of the obstacles. The purpose of this part is an evaluation of HDAO and HDPO in terms of hole detection. In the next part, we evaluate the performance of Greedy-HHP. We analyze the healing provided by Greedy-HHP for a different number of sensors and holes.

\subsubsection{HDAO and HDPO}
	We simulate HDAO and HDPO over an RoI of $100m \times 80m$. According to Figure~\ref{hole-without-obstacle}, after deploying 23 sensors in the environment, HDAO obtains the holes and reports the boundary edges at a runtime of 5.68 seconds. In case that the environment consists of 34 sensors and three obstacles (Figure~\ref{hole-with-obstacle}), HDPO calculates the intersection points between obstacles and sensors and obtains the holes at a runtime of 24.29 seconds. The simulation results show that HDAO and HDPO detect the holes truly in the low and high density of the sensors, and the type of created holes (open and close) do not have effects on the algorithms' performance. Also, we observe that the shape of the obstacles has no impact on the algorithm's precision (see Figure~\ref{hole-with-obstacle}). 
	\begin{figure}[H]
	\centering
	\begin{subfigure}[b]{.5\linewidth}
		\centering
		\includegraphics[scale=.28]{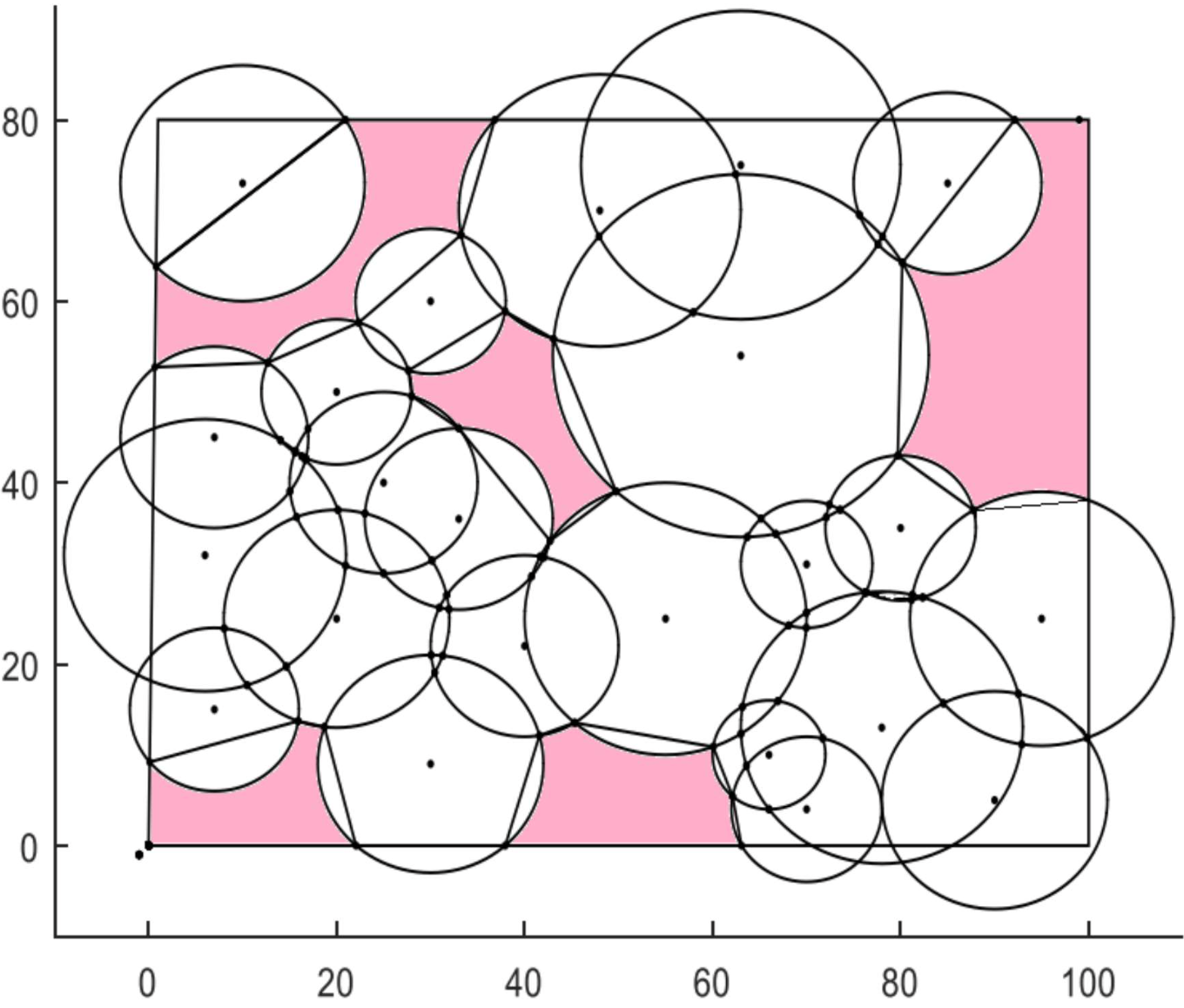}
		\caption{}
			\label{hole-without-obstacle} 
	\end{subfigure}
	\begin{subfigure}[b]{.45\linewidth}               
		\centering
		\includegraphics[scale=.4]{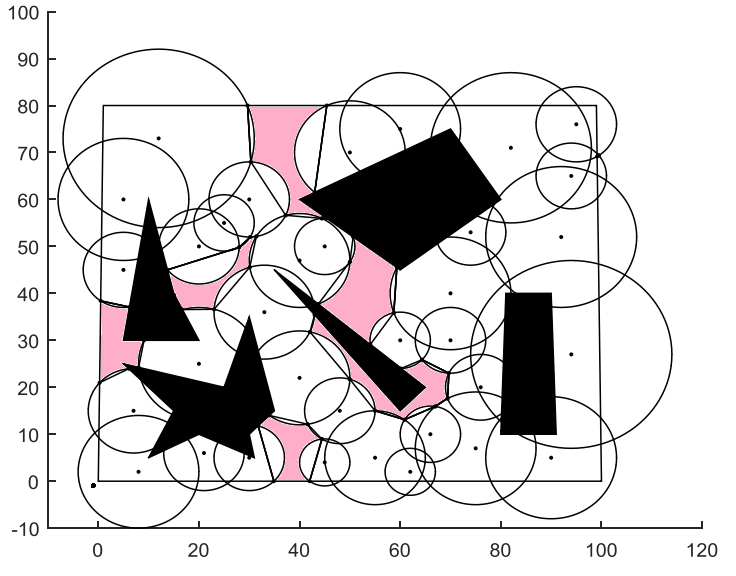}
		\caption{}
	\label{hole-with-obstacle} 
	\end{subfigure}
	\caption{(a) Hole detection in an environment without obstacles, gray areas are the holes detected by HDAO. (b) Hole detection in an environment with obstacles, black areas are the obstacles.} 
	\label{hole-with-without-obstacle} 
\end{figure}
	
	The main advantages of our proposed algorithms in comparison with the previous ones are completeness, accuracy and low time complexity. In case that the environment contains no obstacle and the sensing range of the sensors is different, the time complexity of our algorithm is $O(n \cdot \log^2 n)$, while in the case that the sensing range of the sensors is identical, the time complexity of the previous algorithms in~\cite{kang2013detection,sahoo2010vector} and~\cite{li2015coverage} are $O(n^3)$ and $O(n^2)$, respectively. Further, the proposed algorithms in~\cite{kang2013detection,sahoo2010vector}, after obtaining boundary points of the holes, do not detect the hole boundary edges precisely. 
	
\subsubsection{Greedy-HHP }
	In this subsection, we simulate Greedy-HHP for a set of sensors, including both static and mobile sensors.
	%by varying the sensing range of the sensors and the number of holes in an environment
	%for different numbers and the different sizes of the sensors and holes in an environment.
	We analyze the behaviour of the algorithm before and after moving the mobile sensors. In case that the environment contains no obstacle, after running the algorithm, we can see that mobile sensors move toward the holes and enhance coverage in the environment. Some results for hole healing by considering the different numbers and the different sizes of the sensors and holes are depicted in Figure~\ref{healing-without-obstacle}.
	\begin{figure}[H]
		\centering
		\begin{subfigure}[b]{.5\linewidth}
			\centering
			\includegraphics[scale=.5]{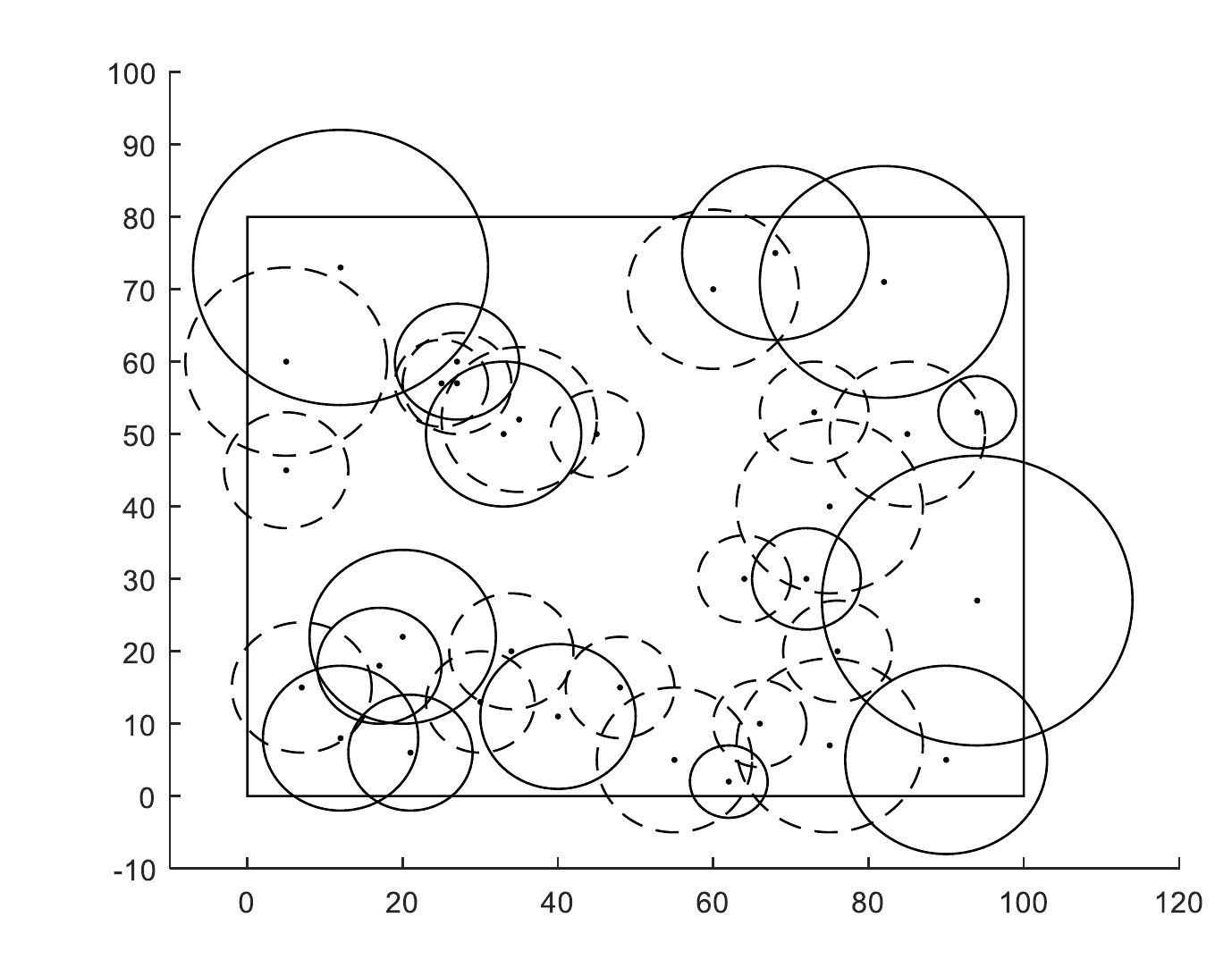}
			%\caption{}
			%\label{hole-without-obstacle} 
		\end{subfigure}
		\begin{subfigure}[b]{.45\linewidth}               
			\centering
			\includegraphics[scale=.5]{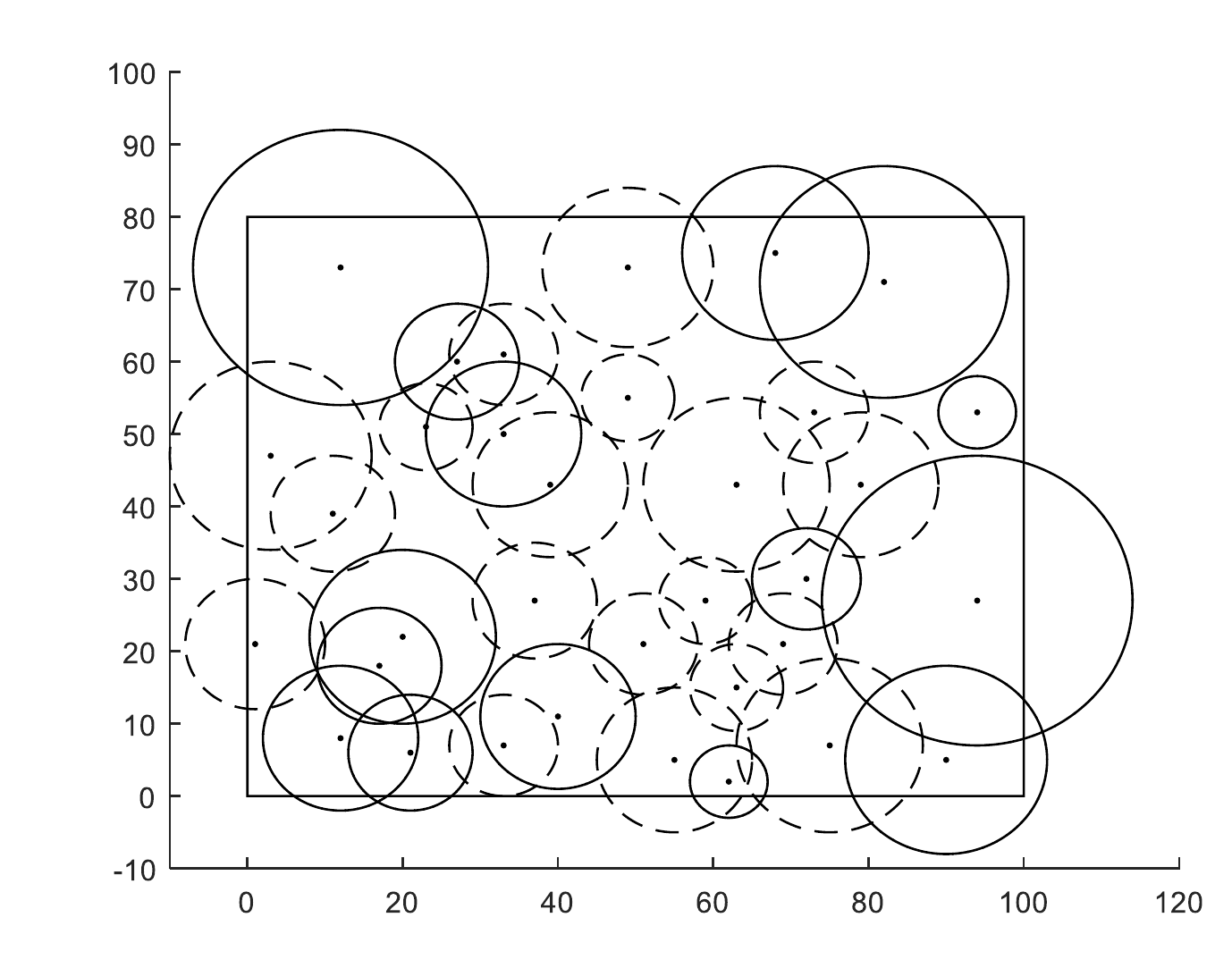}
			%\caption{}
			%\label{hole-with-obstacle} 
		\end{subfigure}
			\begin{subfigure}[b]{.5\linewidth}
			\centering
			\includegraphics[scale=.5]{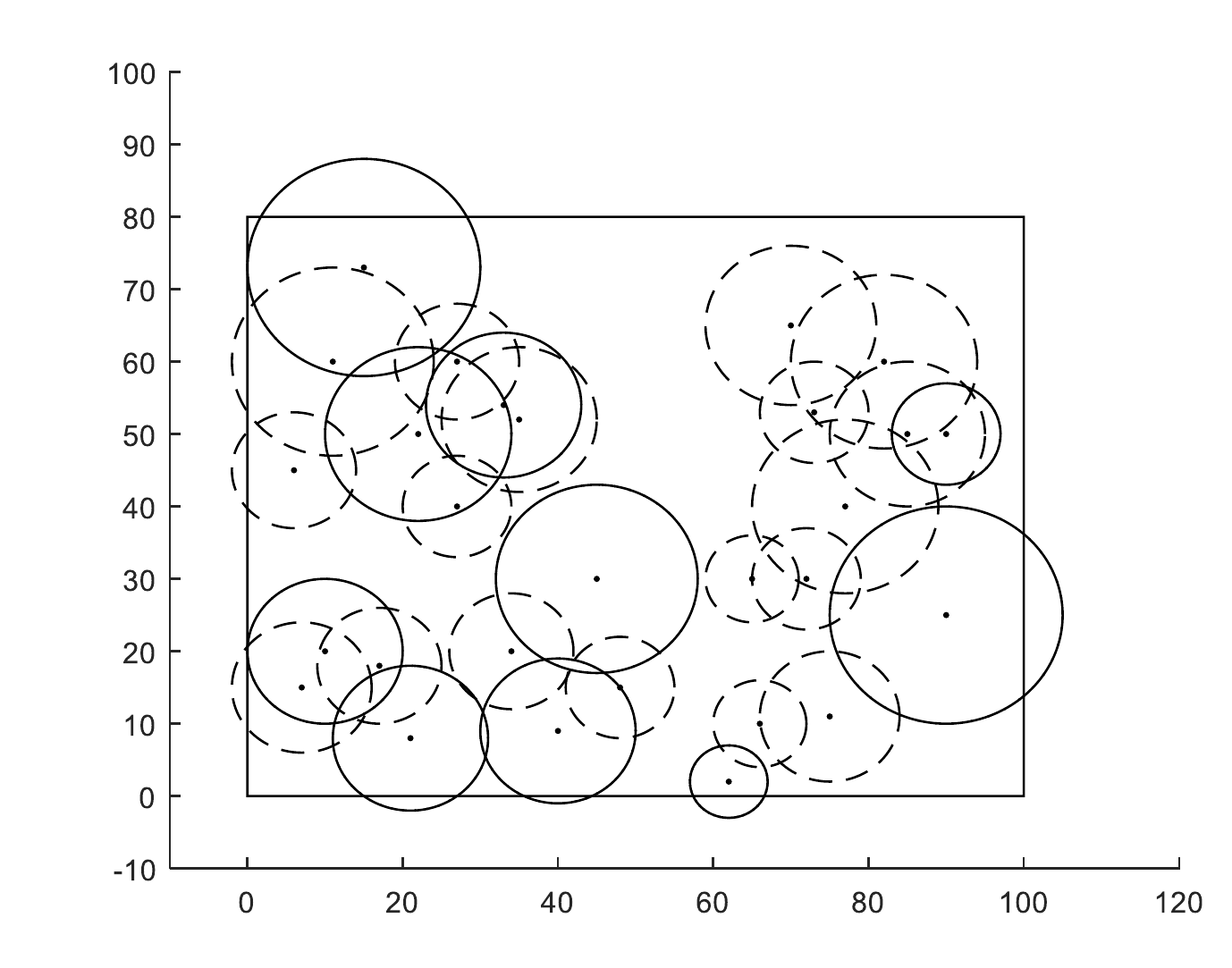}
			\caption{Before running Greedy-HHP.}
			%\label{hole-without-obstacle} 
		\end{subfigure}
		\begin{subfigure}[b]{.45\linewidth}               
			\centering
			\includegraphics[scale=.5]{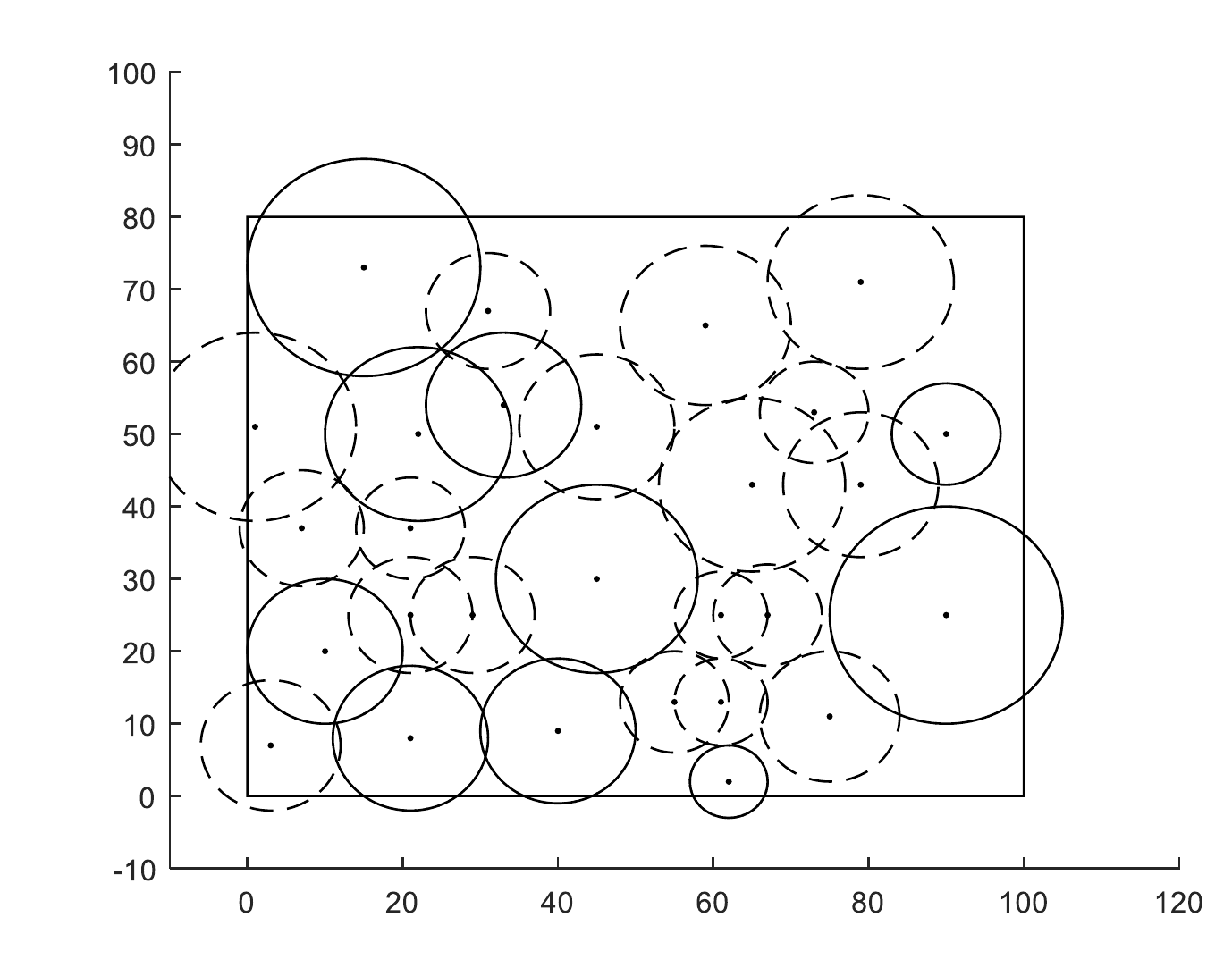}
			\caption{After running Greedy-HHP.}
			%\label{hole-with-obstacle} 
		\end{subfigure}
		\begin{subfigure}[b]{.45\linewidth}               
			\centering
			\includegraphics[scale=.5]{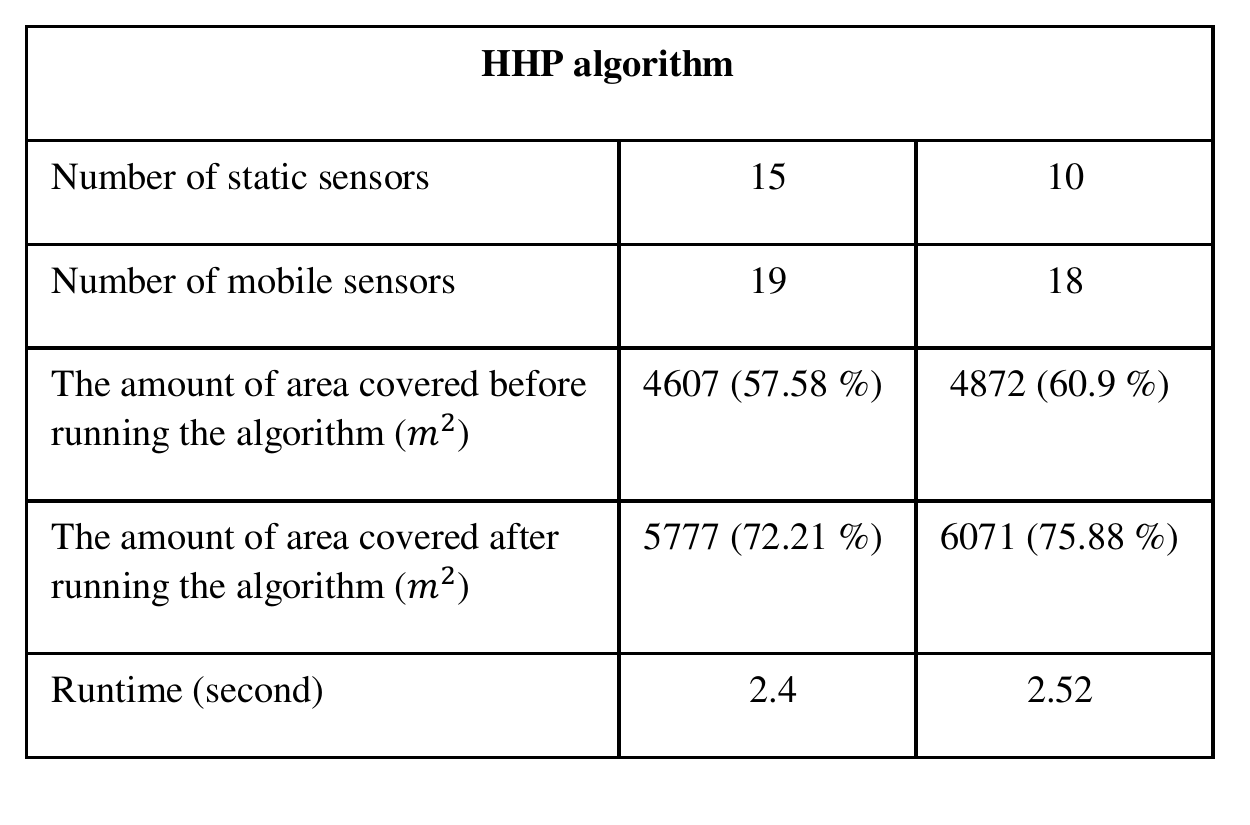}
			%\caption{}
			%\label{hole-with-obstacle} 
   		\end{subfigure}
	
		\caption{Hole healing in an environment without obstacles, the mobile sensors are shown by dashed circles.} 
		\label{healing-without-obstacle} 
	\end{figure}
	In order to further explore the efficiency of Greedy-HHP, we report the percentage of hole recovered area before and after applying Greedy-HHP over the RoI in the presence of the obstacles (see Figure~\ref{healing-with-obstacle}).
	It is observed that the percentage of hole recovery is increased by moving mobile sensors. Furthermore, when the sensing range of the static sensors is small, mobile sensors are able to cover maximum amount of holes.
	
	The simulation results show that the total coverage in the environment is increased after running the algorithms, and the overlapping area is decreased after moving the mobile sensors significantly. Further, by increasing the number of mobile sensors and their sensing range, the coverage area's percentage is improved.
\begin{figure}[H]
	\centering
	\begin{subfigure}[b]{.5\linewidth}
		\centering
		\includegraphics[scale=.5]{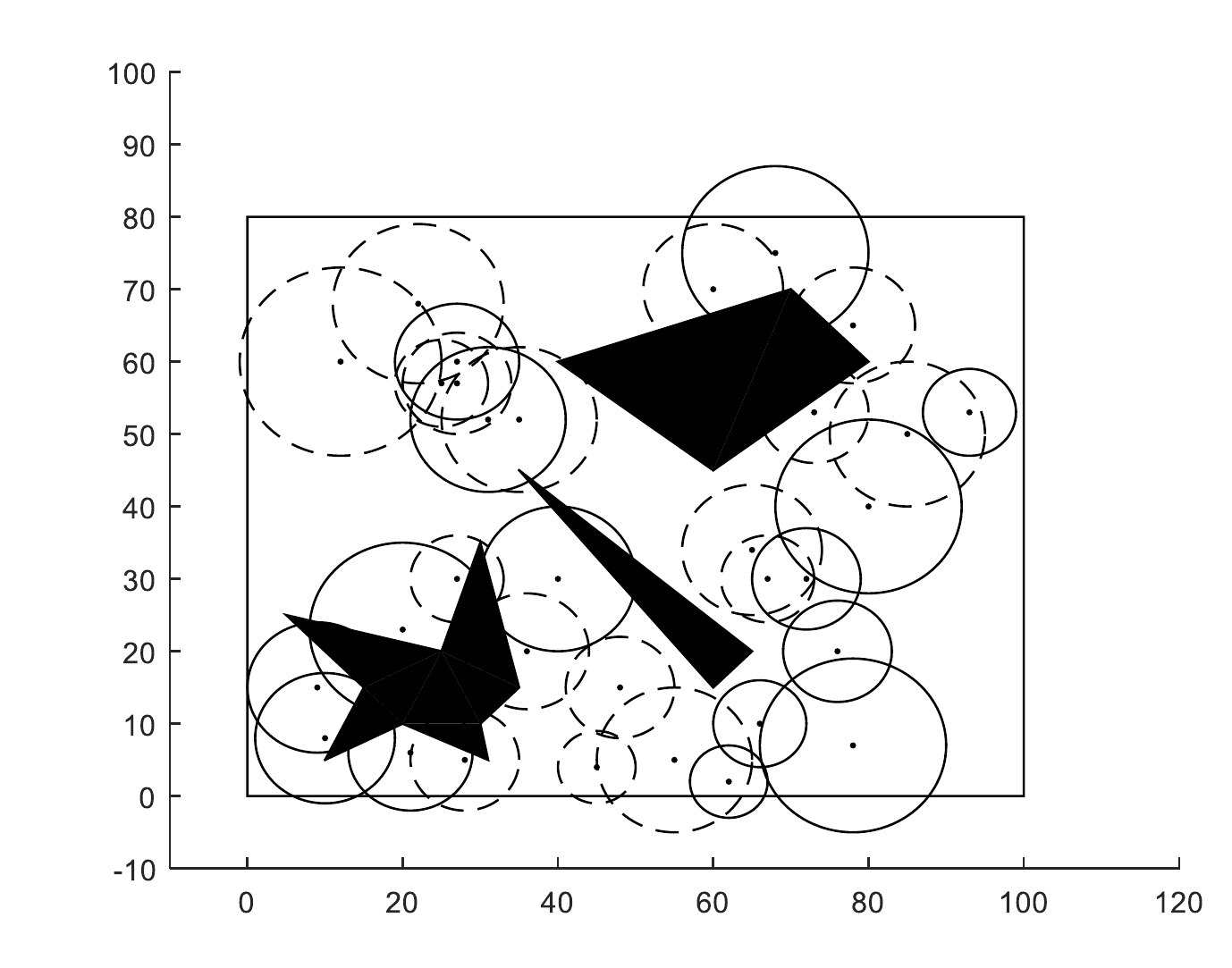}
		%\caption{}
		%\label{hole-without-obstacle} 
	\end{subfigure}
	\begin{subfigure}[b]{.45\linewidth}               
		\centering
		\includegraphics[scale=.5]{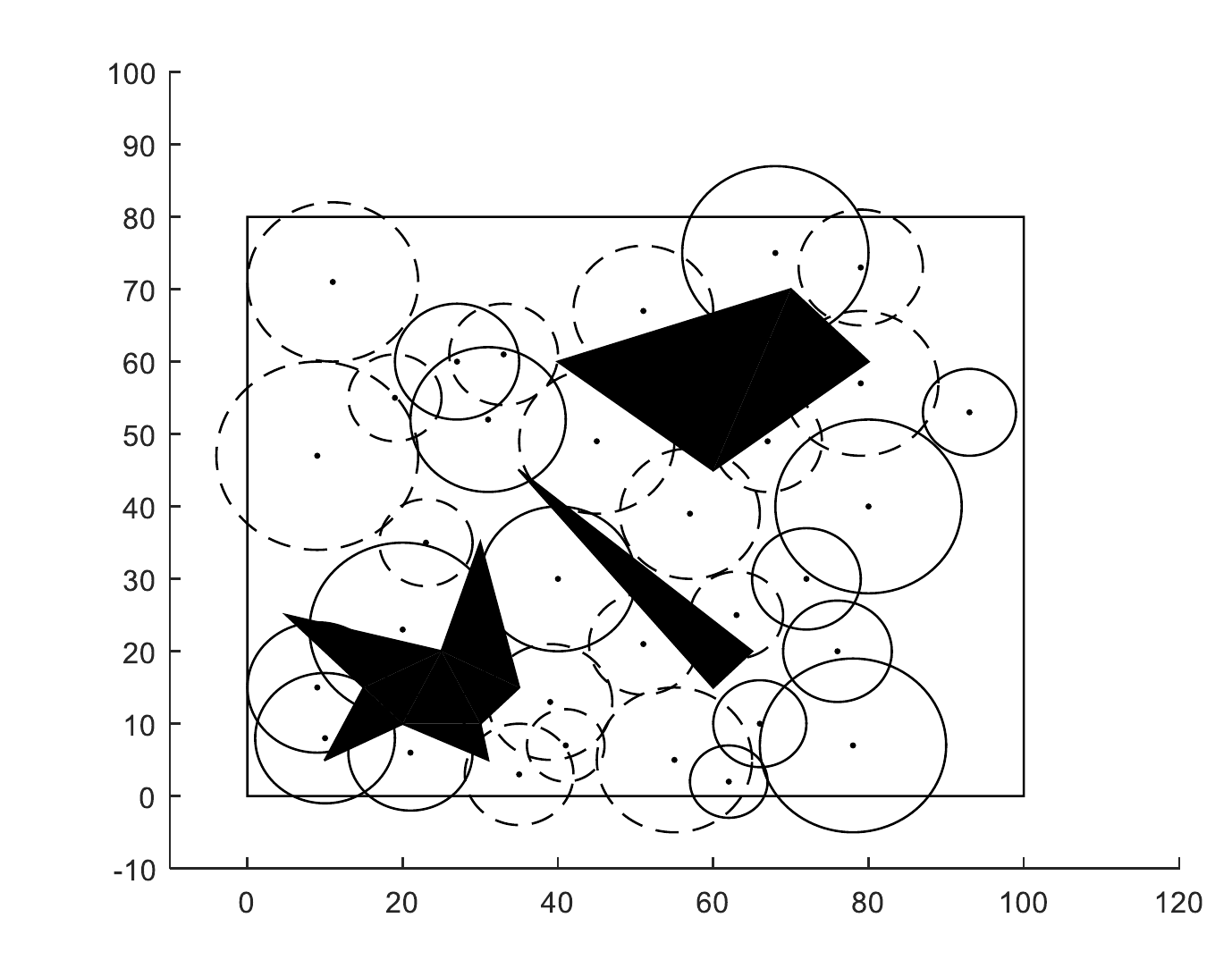}
		%\caption{}
		%\label{hole-with-obstacle} 
	\end{subfigure}
	\begin{subfigure}[b]{.5\linewidth}
		\centering
		\includegraphics[scale=.5]{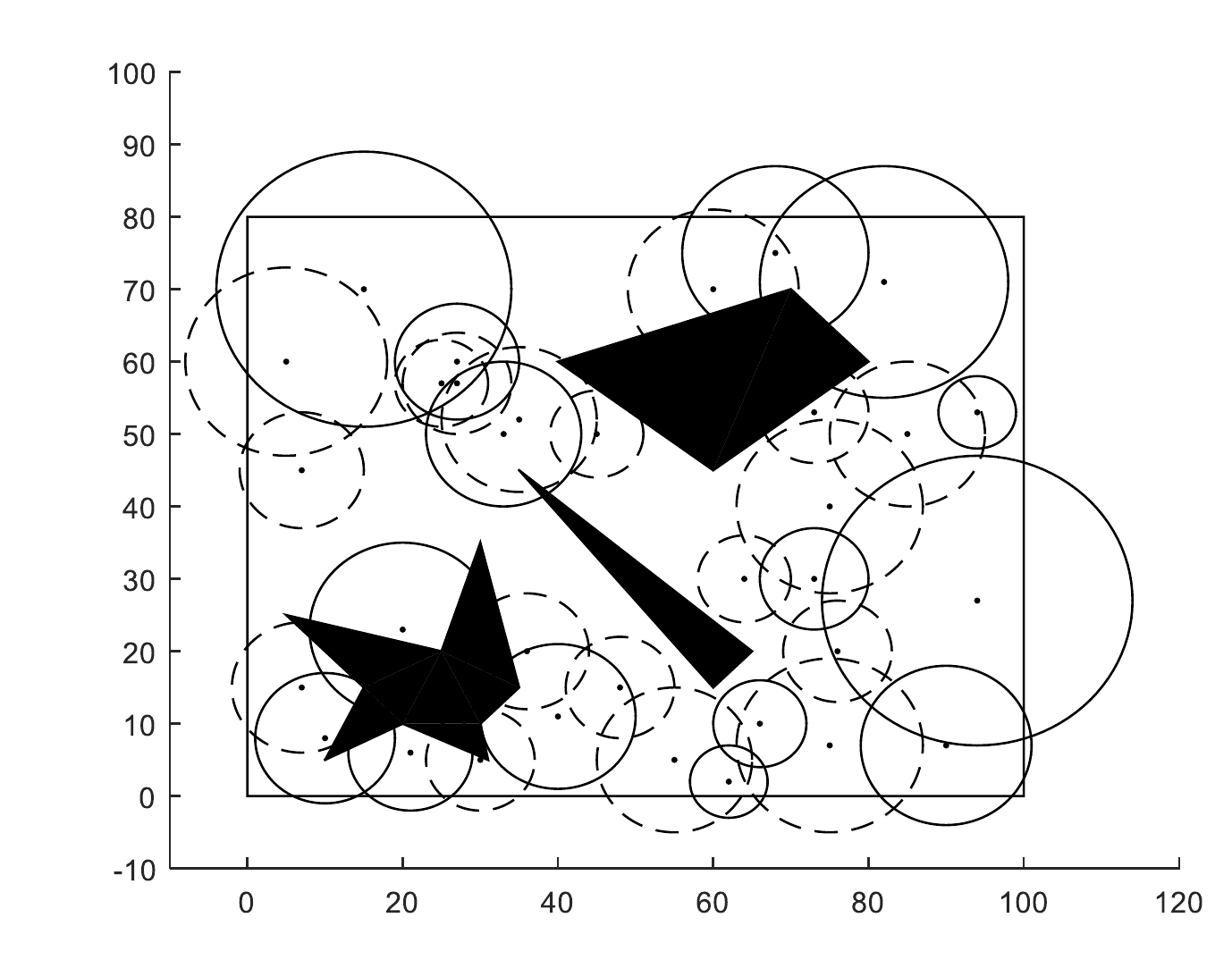}
		\caption{Before running Greedy-HHP.}
		%\label{hole-without-obstacle} 
	\end{subfigure}
	\begin{subfigure}[b]{.45\linewidth}               
		\centering
		\includegraphics[scale=.5]{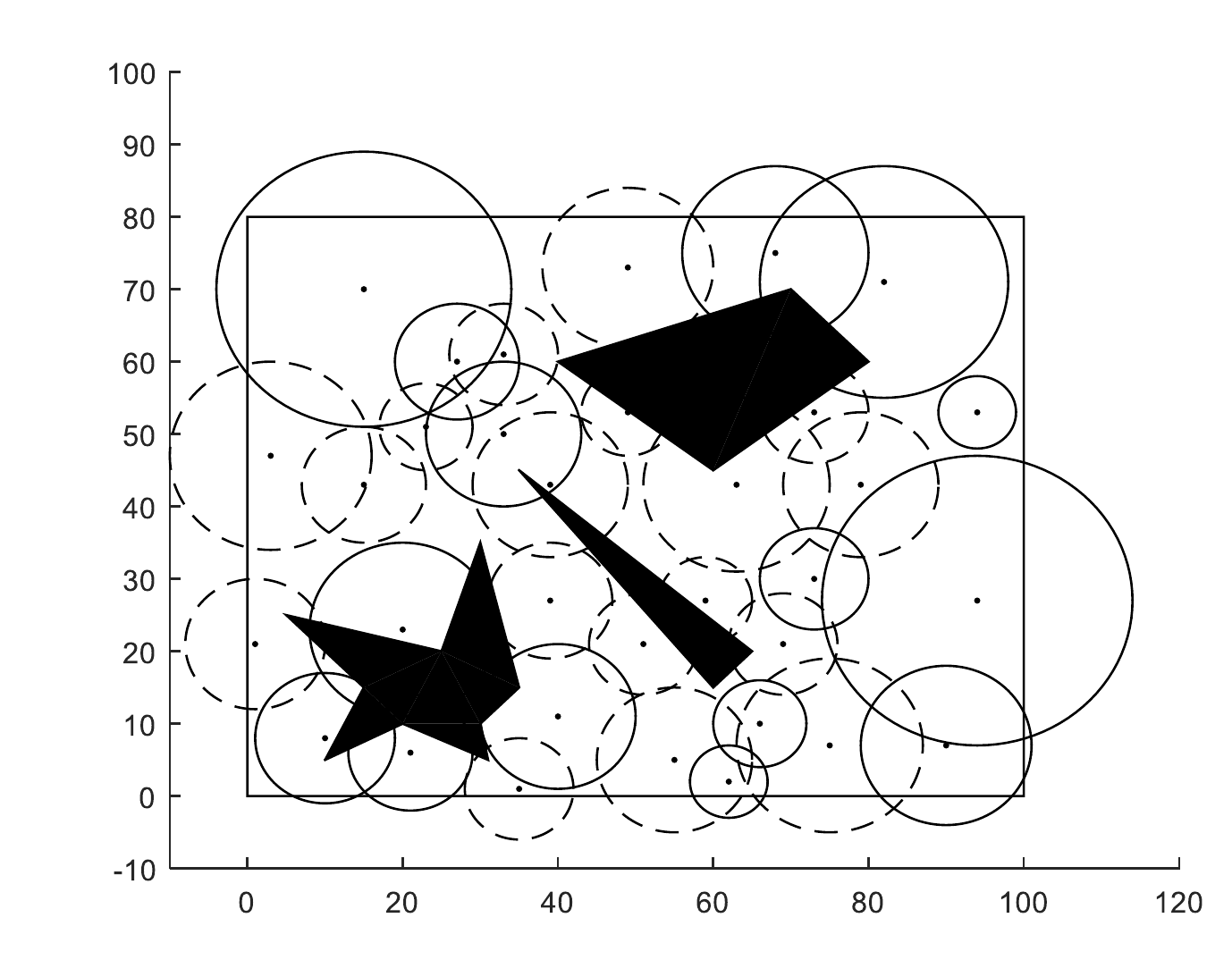}
		\caption{After running Greedy-HHP.}
		%\label{hole-with-obstacle} 
	\end{subfigure}
	\begin{subfigure}[b]{.45\linewidth}               
		\centering
		\includegraphics[scale=.5]{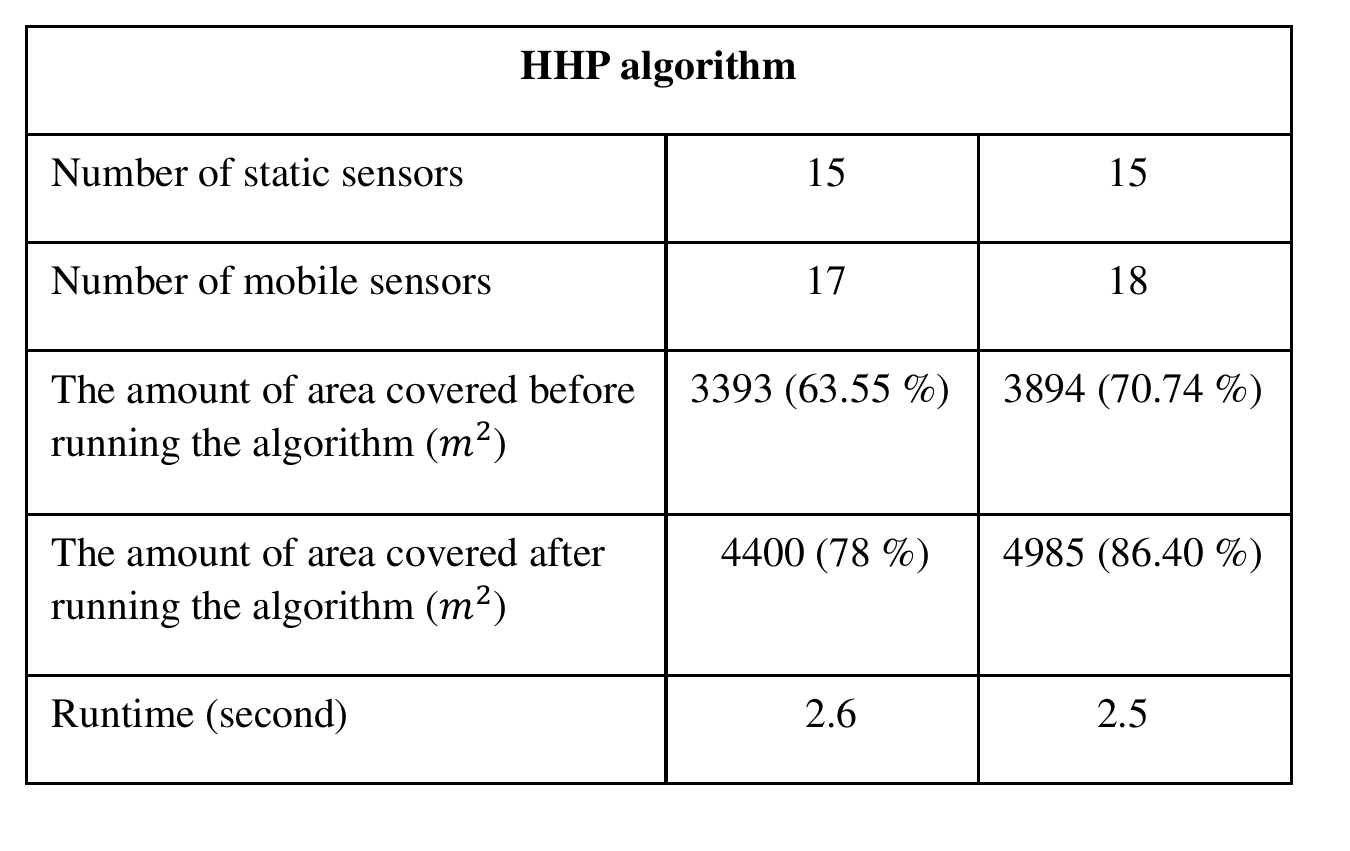}
		%\caption{}
		%\label{hole-with-obstacle} 
	\end{subfigure}
	
	\caption{Hole healing in an environment with obstacles, the mobile sensors are shown by dashed circles.} 
	\label{healing-with-obstacle} 
\end{figure}

\section{Conclusion}
In this paper, we investigated the problems of hole detection and healing in an environment of hybrid sensors with different sensing ranges. We presented a polynomial-time algorithm to detect all holes precisely in the presence of obstacles in the environment. %Detection of coverage holes is achieved by using the AWVD. In the first case, the environment contains no obstacle, we show that the total number of boundary points of the holes is $O(n)$, where $n$ is the number of sensors, and propose a centralized algorithm for detecting and computing the area of the holes with the time complexity of $O(n \cdot \log^2 n)$. In the second case, the environment consists of a set of obstacles; we show that $\Omega(n \cdot z)$, where $z$ is the total number of vertices of the obstacles, is a lower bound for every hole detection algorithm. In this case, we present a centralized algorithm with the time complexity of $O(n (\log^2 n + l \cdot z))$, where $l$ is the number of obstacles. 
We then proposed a $1/2$-approximation algorithm to maximize the coverage area by moving the mobile sensors. The simulation results demonstrated that our proposed algorithms cover the holes in the environment efficiently. % solution in hybrid WSNs. %Compared with similar methods, our hole detection algorithms (HDAO and HDPO) discover the hole boundaries more accurately. 
Designing algorithms that run faster and provide a better approximation factor would be a worthwhile contribution.

%Several directions of future work are conceivable. First, designing algorithms that would provide a better approximation factor would be a worthwhile contribution. Second, it would be interesting to design heuristics algorithms to cover the holes efficiently.

% \input{src/appendix}

\bibliography{./references}
\end{document}